\documentclass[11pt]{article}
\usepackage{graphicx}
\usepackage{gensymb}
\usepackage{amsmath}
\usepackage{algpseudocode}

\usepackage{algorithm}
\usepackage[margin=1in]{geometry}
\usepackage[utf8]{inputenc}
\usepackage{amsthm}
\newtheorem{theorem}{Theorem}
\newtheorem{remark}{Remark}
\newtheorem{definition}{Definition}

\newtheorem{lemma}[theorem]{Lemma}
\usepackage{color}
\usepackage{tikz,comment}
\usepackage{times}
\usepackage[nottoc]{tocbibind}
\usepackage{t1enc}
\usepackage{graphicx}
\usepackage{textcomp}
\usepackage{color,xcolor}
\usepackage{epstopdf}
\usepackage{rotating}
\usepackage{xifthen}
\usepackage{paralist}
\usepackage[nointegrals]{wasysym}
\usepackage{multirow}
\usepackage[none]{hyphenat}
\usepackage{caption} 
\usepackage{microtype}
\date{}
\newcommand{\mycommfont}[1]{\textcolor{blue}{#1}} % Example: Blue color

\algrenewcommand\algorithmiccomment[1]{\hfill\mycommfont{#1}}

\AtBeginDocument{%
  }

\begin{document}

%\title{Efficient Robot Dispersion on Unoriented Grids \\ with and without Faulty Robots}

\title{Optimizing Robot Dispersion on Grids: with and without Fault Tolerance}

\author{Rik Banerjee \footnote{Indian Statistical Institute, Kolkata, India, {\em rikarjya@gmail.com}}, Manish Kumar\footnote{Indian Institute of Technology, Madras, India, {\em manishsky27@gmail.com}}, and Anisur Rahaman Molla\footnote{Indian Statistical Institute, Kolkata, India, {\em anisurpm@gmail.com}}}
% \institute{IIT Madras\\
% \email{manishsky27@gmail.com}
% }
\maketitle              
%%%%%%%%%%%%%%%%%%%%%%%%%%%%%%%%%%%%%%%%%%%%%%%%%%%%%%

\begin{abstract}
The introduction and study of dispersing mobile robots across the nodes of an anonymous graph have recently gained traction and have been explored within various graph classes and settings. While optimal dispersion solution was established for {\em oriented} grids [Kshemkalyani et al., WALCOM 2020], a significant unresolved question pertains to whether achieving optimal dispersion is feasible on an {\em unoriented} grid. This paper investigates the dispersion problem on unoriented grids, considering both non-faulty and faulty robots. The challenge posed by unoriented grids lies in the absence of a clear sense of direction for a single robot moving between nodes, as opposed to the straightforward navigation of oriented grids.  

%\shortOnly{
We present three deterministic algorithms tailored to our robot model. The first and second algorithms deal with the dispersion of faulty and non-faulty robots, ensuring both time and memory optimization in oriented and unoriented grids, respectively. Faulty robots that are prone to crashing at any time, causing permanent failure. In both settings, we achieve dispersion in $O(\sqrt{n})$ rounds while requiring $O(\log n)$ bits of memory per robot. The third algorithm tackles faulty robots prone to crash faults in an unoriented grid. In this scenario, our algorithm operates within $O(\sqrt{n} \log n)$ time and uses $O(\sqrt{n} \log n)$ bits of memory per robot. The robots need to know the value of $n$ for termination.
%} 
\end{abstract}

\section{Introduction}\label{sec: introduction}
% \manish{
% \begin{itemize}
%     \item Dada, Have a look at the Theorem 2. Should we mention the number of robots in $k$ since it is not reflected in the complexity?
%     \item 5 pages and around 700 words are extra.
%     \item Should we remove the Related work part?
% \end{itemize}
% }
The distribution of autonomous mobile robots for achieving coverage across an area is a highly pertinent challenge within distributed robotics, as highlighted in \cite{HABFM02,HABFM03}. More recently, this issue has been framed in the context of graphs in the following manner: In a scenario where $k$ robots are initially situated on the nodes of an $n$-node graph, the robots undertake autonomous repositioning to achieve a final configuration wherein each robot occupies a distinct graph node (referred to as the {\em dispersion} problem) \cite{AM18}. This problem holds practical significance across various applications, such as the repositioning of self-driving electric cars (analogous to robots) to available charging stations (equivalent to nodes). This assumption involves the cars utilizing intelligent communication methods to locate unoccupied charging stations \cite{AM18,KA19}. Furthermore, the problem's importance stems from its interconnectedness with numerous other extensively researched challenges in autonomous robot coordination, including exploration, scattering, and load balancing \cite{AM18,KA19}.

The dispersion of mobile robots has garnered attention across various graph classes, including trees\cite{AM18,KA19}, rings\cite{AM18, KA19, MMM21}, arbitrary graphs \cite{AM18, KMS19, KMS20, KMS22, KS21, SSKM20}, dynamic graphs \cite{KMS20-dynamic}, directed graphs \cite{IPS22}. In the grid graph, the problem was explored by Kshemkalyani et al. \cite{KMS20}, but they considered oriented grid (called planar grid) and non-faulty robots exclusively. Orientation plays a pivotal role in the symmetric graphs, Barrière et. al. studied the scattering of autonomous mobile robots in the grid \cite{barriere2011uniform} and Becha et al. constructed a sense of direction by mobile robots in a torus \cite{becha2007optimal}.

{\bf Oriented vs Unoriented Grid:} In an oriented grid (see Figure~\ref{fig:planar_grid}), the ports are organized in such a way that allows a single robot to traverse the grid along a path with a clear sense of direction. When a robot enters a node via an incoming port, it simply needs to select the second port (i.e., leave one port after the incoming port and select the next port) as its outgoing port to continue moving in the same direction. In Figure~\ref{fig:planar_grid}, as a robot enters from port~1 at node $u$, it has the sense that the straight path leads from port~3 to node $v$. However, this straightforward approach is not applicable in an unoriented grid where the ports are interconnected arbitrarily, as depicted in Figure~\ref{fig:non_planar_grid}. In such a scenario, robots cannot distinguish whether they are moving in the same direction (i.e., along a row or column) or traversing in a cycle or zigzag manner across the unoriented grid. In Figure~\ref {fig:non_planar_grid}, a robot enters from the port~1 at node $x$. It is tough to decide based on the edges which path leads in the straight direction to $y$. In the grid, it appears that port~4 leads to the $y$ unlike port~3.
Consequently, it's not feasible to adapt the algorithm proposed in \cite{KMS20} to work in an unoriented grid.

% \begin{figure}[htbp]
%     \centering
%     \begin{minipage}{0.45\textwidth}
%         \centering
%         \begin{tikzpicture}[scale=1.3]
%   % Nodes
%   \foreach \x in {0,1,2,3}
%     \foreach \y in {0,1,2,3}
%       \node[draw,circle,fill=black,inner sep=0.07cm] at (\x,\y) (n\x\y) {};

%   % Edges
%   \foreach \x in {0,1,2,3}
%     \foreach \y in {0,1,2}
%       \draw (n\x\y) -- (n\x\the\numexpr\y+1\relax);
%   \foreach \x in {0,1,2}
%     \foreach \y in {0,1,2,3}
%       \draw (n\x\y) -- (n\the\numexpr\x+1\relax\y);
% \end{tikzpicture}
%     \caption{$16$ nodes planar square grid.}
%     \label{fig:planar_grid}
%     \end{minipage}
%     \hfill
%     \begin{minipage}{0.45\textwidth}
%         \centering
%         \begin{tikzpicture}[scale=1.1]
%   % Nodes
%   \foreach \x in {0,1,2,3}
%     \foreach \y in {0,1,2,3}
%       \node[draw,circle,fill=black,inner sep=0.07cm] at (\x,\y) (n\x\y) {};

%   % Function to generate random angle between -10 and 10
%   \newcommand{\randomangle}{rand*360}

%   % Edges with curves
%   \foreach \x in {0,1,2,3}{
%     \foreach \y in {0,1,2}{
%       \draw (n\x\y) to[out=\randomangle, in=180-\randomangle] (n\x\the\numexpr\y+1\relax);
%     }
%   }

%   \foreach \x in {0,1,2}{
%     \foreach \y in {0,1,2,3}{
%       \draw (n\x\y) to[out=\randomangle, in=\randomangle] (n\the\numexpr\x+1\relax\y);
%     }
%   }
% \end{tikzpicture}
%     \caption{$16$ nodes non-planar square grid.}
%     \label{fig:non_planar_grid}
%     \end{minipage}
% \end{figure}

\begin{figure}[htbp]
    \centering
    \begin{minipage}{0.45\textwidth}
        \centering
\begin{tikzpicture}[scale=1.3]
    % Nodes with labels
    \node[draw,circle,fill=blue,inner sep=0.07cm] (A1) at (0,0) {};
    \node[draw,circle,fill=blue,inner sep=0.07cm] (A2) at (1,0) {};
    \node[draw,circle,fill=blue,inner sep=0.07cm] (A3) at (2,0) {};
    \node[draw,circle,fill=blue,inner sep=0.07cm] (A4) at (3,0) {};
    
    \node[draw,circle,fill=blue,inner sep=0.07cm] (B1) at (0,1) {};
    \node[draw,circle,fill=blue,inner sep=0.07cm,label=below:$u$] (B2) at (1,1) {};
    \node[draw,circle,fill=blue,inner sep=0.07cm,label=below:$v$] (B3) at (2,1) {};
    \node[draw,circle,fill=blue,inner sep=0.07cm] (B4) at (3,1) {};
    
    \node[draw,circle,fill=blue,inner sep=0.07cm] (C1) at (0,2) {};
    \node[draw,circle,fill=blue,inner sep=0.07cm] (C2) at (1,2) {};
    \node[draw,circle,fill=blue,inner sep=0.07cm] (C3) at (2,2) {};
    \node[draw,circle,fill=blue,inner sep=0.07cm] (C4) at (3,2) {};
    
    \node[draw,circle,fill=blue,inner sep=0.07cm] (D1) at (0,3) {};
    \node[draw,circle,fill=blue,inner sep=0.07cm] (D2) at (1,3) {};
    \node[draw,circle,fill=blue,inner sep=0.07cm] (D3) at (2,3) {};
    \node[draw,circle,fill=blue,inner sep=0.07cm] (D4) at (3,3) {};
    
    % Horizontal edges
    \draw (A1) -- (A2);
    \draw (A2) -- (A3);
    \draw (A3) -- (A4);
    
    \draw (B1) -- (B2);
    \draw (B2) -- node[right = 10, above,red]{1}(B3);
    \draw (B3) -- (B4);

    \draw (B1) --node[right = 10, above,red] {1} (B2);
    \draw (B2) node[right = 10, above,red]{3}--(B3);
    \draw (B3) node[right = 10, above,red]{3}-- (B4);
    
    \draw (C1) -- (C2);
    \draw (C2) -- (C3);
    \draw (C3) -- (C4);
    
    \draw (D1) -- (D2);
    \draw (D2) -- (D3);
    \draw (D3) -- (D4);
    
    % Vertical edges
    \draw (A1) -- (B1);
    \draw (A2) -- (B2);
    \draw (A3)node[right=5, above=16,red]{2} -- (B3);
    \draw (A4) -- (B4);

    \draw (A2) node[right=5, above=16,red]{2}-- (B2);
    \draw (A3) -- (B3);
    
    \draw (B1) -- (C1);
    \draw (B2)node[right, above=2, red]{4} -- (C2);
    \draw (B3) node[right, above=2, red]{4}-- (C3);
    \draw (B4) -- (C4);

    \draw (B2) -- (C2);
    \draw (B3) -- (C3);
    
    \draw (C1) -- (D1);
    \draw (C2) -- (D2);
    \draw (C3) -- (D3);
    \draw (C4) -- (D4);
\end{tikzpicture}

    \caption{$16$ nodes oriented square grid.}
    \label{fig:planar_grid}
    \end{minipage}
    \hfill
    \begin{minipage}{0.45\textwidth}
        \centering
       \begin{tikzpicture}[scale=1.1]
  % Nodes
  \node[draw,circle,fill=blue,inner sep=0.07cm] (n00) at (0,0) {};
  \node[draw,circle,fill=blue,inner sep=0.07cm] (n01) at (0,1) {};
  \node[draw,circle,fill=blue,inner sep=0.07cm] (n02) at (0,2) {};
  \node[draw,circle,fill=blue,inner sep=0.07cm] (n03) at (0,3) {};

  \node[draw,circle,fill=blue,inner sep=0.07cm] (n10) at (1,0) {};
  \node[draw,circle,fill=blue,inner sep=0.07cm,label=below:$x$] (n11) at (1,1) {};
  \node[draw,circle,fill=blue,inner sep=0.07cm] (n12) at (1,2) {};
  \node[draw,circle,fill=blue,inner sep=0.07cm] (n13) at (1,3) {};

  \node[draw,circle,fill=blue,inner sep=0.07cm] (n20) at (2,0) {};
  \node[draw,circle,fill=blue,inner sep=0.07cm,label=below:$y$] (n21) at (2,1) {};
  \node[draw,circle,fill=blue,inner sep=0.07cm] (n22) at (2,2) {};
  \node[draw,circle,fill=blue,inner sep=0.07cm] (n23) at (2,3) {};

  \node[draw,circle,fill=blue,inner sep=0.07cm] (n30) at (3,0) {};
  \node[draw,circle,fill=blue,inner sep=0.07cm] (n31) at (3,1) {};
  \node[draw,circle,fill=blue,inner sep=0.07cm] (n32) at (3,2) {};
  \node[draw,circle,fill=blue,inner sep=0.07cm] (n33) at (3,3) {};

  % Edges with curves
  \draw (n00) to[out=95, in=75] (n01);
  \draw (n00) to[out=225, in=95] (n10);
  \draw (n01) to[out=345, in=155] (n02);
  \draw (n01) to[out=265, in=275] node[right = 10, above,red]{1} (n11);
  \draw (n02) to[out=185, in=195] (n03);
  \draw (n02) to[out=140, in=115] (n12);
  \draw (n13) to[out=345, in=235] (n03);
  \draw (n10) to[out=245, in=355]node[above=10,red]{2} (n11);
  \draw (n10) to[out=145, in=275] (n20);
  \draw (n11) node[right, above =1 ,red]{4} to[out=415, in=295]  (n12);
  \draw (n11) node[right = 5,red]{3} to[out=125, in=215]  (n21);
  \draw (n11) node[right = 20,red]{1} to[out=125, in=215]  (n21);
  \draw (n12) to[out=435, in=235] (n13);
  \draw (n12) to[out=445, in=255] (n22);
  \draw (n13) to[out=455, in=275] (n23);
  \draw (n20) to[out=465, in=295] node[right=5,above=5,red]{2}(n21);
  \draw (n20) to[out=475, in=315] (n30);
  \draw (n21) node[left =2, above =3 ,red]{4}to[out=485, in=325] (n22);
  \draw (n21) node[right = 5, above,red]{3}to[out=145, in=345] (n31);
  \draw (n22) to[out=245, in=365] (n23);
  \draw (n22) to[out=345, in=385] (n32);
  \draw (n23) to[out=445, in=405] (n33);
  \draw (n30) to[out=45, in=425] (n31);
  \draw (n31) to[out=145, in=445] (n32);
  \draw (n32) to[out=245, in=4655] (n33);
\end{tikzpicture}

    \caption{$16$ nodes unoriented square grid.}
    \label{fig:non_planar_grid}
    \end{minipage}
\end{figure}

We provide three novel deterministic algorithms for dispersion on oriented and unoriented grid graphs. Our first algorithm works with faulty robots in an oriented grid graph. The second algorithm works with non-faulty robots and the third algorithm works with faulty robots in an unoriented grid graph, in which, a faulty robot can crash at any time during the execution of the protocol and never respond after crashing. Among the best-known results on dispersion \cite{KS21, CKMS23}, the paper \cite{KS21} considered non-faulty robots, and \cite{CKMS23} considered faulty robots. However, applying the arbitrary graph results of \cite{KS21,CKMS23} to grid graphs yields memory-optimal solutions but not the time-optimal ones. Our results and a comparison with the closely related work \cite{KMS20} are given in Table~\ref{tab:results} for reference.

\begin{table}%[H]

\begin{tabular}{ |p{4.50cm}|p{3cm}|p{1.5cm}|p{2.5cm}|p{3cm}|}
%\hline
%\multicolumn{5}{|c|}{\sc Comparative Analysis of dispersion on a Grid} \\
\hline
Algorithm & Grid Type & Faulty & Time & Memory (in bits)\\
\hline
 Kshemkalyani et al. \cite{KMS20}& Oriented & No & $O(\text{min}(k,\sqrt{n}))$ & $O(\log k)$\\
 %\cite{KMS20}   & Planar & Yes & \\
    \textbf{Algorithm in Section~\ref{sec: faulty-planar} }& Oriented & Yes & $O(\sqrt{n})$ & $O(\log n)$\\
    \textbf{Algorithm in Section~\ref{sec: honest_robot}}& Unoriented & No &$O(\sqrt{n})$ & $O(\log n)$\\
    \textbf{Algorithm in Section~\ref{sec: crash-fault non-planar}}& Unoriented & Yes & $O(\sqrt{n}\log n)$ & $O(\sqrt{n} \log n)$\\
\hline
\end{tabular}\caption{Results for the dispersion of $k \leq n$ robots and up to $f\leq n$ faulty robots on an $n$-node square grid.}
    \label{tab:results}
\end{table}

\noindent \textbf{Our Contributions.} We present algorithms for the dispersion of $k$ robots in any arbitrary anonymous square grid in three different models, namely, faulty robots' dispersion in oriented grid, non-faulty robots' dispersion in unoriented grid, and faulty robots' dispersion in unoriented grid.  All these algorithms are deterministic. Specifically, we have the following three contributions when the square grid size, $n$, is known.
\begin{enumerate}
    \item \textbf{Faulty robots' dispersion on an oriented grid:} We present dispersion of $k$ robots in an oriented square grid of $n$ nodes having $f$ faulty robots such that $f \leq k \leq n$ that terminates in $O(\sqrt{n})$ rounds and uses memory bits $O(\log n)$ at each robot.
    
    \item \textbf{Non-faulty robots' dispersion on an unoriented grid:} We propose dispersion of $k$ robots in an unoriented square grid of $n$ nodes such that $k \leq n$ that terminates in $O(\sqrt{n})$ rounds and uses $O(\log n)$ bits of memory at each robot.
    
    \item \textbf{Faulty robots' dispersion on an unoriented grid:} We present dispersion of $k$ robots in an unoriented square grid of $n$ nodes in the presence of any number of faulty robots that terminates in $O(\sqrt{n}\log n)$ rounds and uses $O(\sqrt{n}\log n)$ bits of memory at each robot.
\end{enumerate}

 Notice that $\Omega(\sqrt{n})$ is the trivial lower bound for round complexity. A robot requires the diameter of the graph time (round) to reach from one end to the other, where $2\sqrt{n}$ is the diameter of the grid. Therefore, our first two results are time optimal, while the third result is time optimal up to $\log n$ factors. While $\Omega(\log k)$ is the lower bound for each robot's memory (in bits)\cite{AM18}. 

\noindent \textbf{Challenges and Techniques.} Recall that robots in the oriented grid possess a sense of direction. Firstly, a robot can reach the boundary nodes (nodes with degree $3$) from the internal of the grid by moving in a straight direction. But if the grid is unoriented in that case the robots have no sense of the direction, therefore, a robot can not reach the boundary in optimal time. Secondly, a robot can reach the corner of the grid by following the boundary. In this way, robots can find each other by sending a single robot to find the appropriate corner to gather at a single corner.  On the other hand, if the robots are faulty then gathering at a single corner becomes challenging due to the faulty nature of the robots. The single robot can crash and corner robots might wait for an indefinite time.

We overcome these challenges as follows. In the non-faulty unoriented setup, two or more robots ensure their direction of movement in a straight line. Firstly, the group of robots ($2$ or more) moves in the direction of the minimum port number. Now the challenge is to ensure the movement in a straight direction. In this, all robots except one remain stationary, and one robot $r_u$ explores the shortest possible ways to reach its initial position in the grid where the rest of the robots are placed. The robot $r_u$ reaches the initial position in two ways, using the shortest path. These two paths lie at $180\degree$ to each other. This helps to determine the direction of the robots in a straight line. The detailed description is given in Section~\ref{sec: honest_robot}. A similar problem is tackled in the faulty unoriented setup by sending half of the robots to explore the other two paths to reach their initial position (detailed description is given in Section~\ref{sec: crash-fault non-planar}). In this way, robots reach the corner with the help of the boundary. Secondly, the faulty corner robots find each other by sending half of their available corner robots in one direction and reach their initial position after a round trip of the grid. By using this strategy, either half of the robots crash or find the appropriate corner (based on ID) to move. In the worst case, every time half of the corner robots may crash which causes an overhead of $O(\log n)$ round trip of the grid and there would not be left any robot to disperse. This problem is optimally tackled in the Section~\ref{sec: faulty-planar} in two different way based on the value of $n$, i.e., even or odd. For odd value of $n$ all the robots gathers at the center of the grid. On the other hand, for even value of $n$, due to unavailability of unique center of grid, robots partition the whole grid in $4$ parts and disperse accordingly. Finding an appropriate node in the unoriented grid might be another challenge after gathering at the corner in a faulty setup since a singleton robot might be left in the internal grid. Doing that effectively is also challenging, we discussed one of the approaches in Section~\ref{sec: crash-fault non-planar}.  

%\vspace{-0.6cm}

\noindent The rest of the paper is organized as follows.

\medskip

\noindent \textbf{Paper Organization.} Section~\ref{sec: model} states our distributed computing model. Section~\ref{sec: related_work} discusses the closely related work. Section~\ref{sec: faulty-planar} discuss the dispersion of faulty robots on oriented grid. In Section~\ref{sec: honest_robot}, the dispersion of non-faulty robots is presented on the unoriented grid. In section~\ref{sec: crash-fault non-planar}, the dispersion of faulty robots on an unoriented grid is discussed. Finally, Section~\ref{sec: conclusion} concludes the paper with some interesting problems.
 
\section{Distributed Computing Model}\label{sec: model}
%insider, non-insider, corner, edge robots and grid specifications.
%consider the grid blocks are at $90\degree$ to each other.
% two adjacent ports are at 180\degree if they are in a straight line
%the graph is non-planar in the sense that it is not possible to find out the angle between adjacent port by observing the port.
% edge of the grid means node with degree 3 while edge of the graph is any edge in the graph connecting two nodes.
% One hop is the distance between any two neighboring nodes. Two nodes are neighboring nodes if there exists a direct path between two nodes without an intermediate node.
%side edge is the edge on the side of the grid with degree 3

\noindent{\bf Graph:} We consider an unweighted, undirected graph $G = (V, E)$ which is a square grid of $n=\sqrt{n}\times \sqrt{n}$ nodes embedded in $2$-dimensional Euclidean plane such that $\mid V \mid = n$ and $\mid E \mid = m$, where $V$ is the set of nodes and $E$ is the set of edges. $G$ is a connected graph having nodes with either degrees $2$ or $3$ or $4$. Nodes with degrees $2$, $3$, and $4$ are considered to be corner nodes, boundary nodes, and internal nodes, respectively. A square grid consists of $4$ corner nodes, $4\sqrt{n}-8$ boundary nodes, and $n-4\sqrt{n}+4$ internal nodes. These nodes are memoryless and resourceless means, unable to store any information and perform computation on them. Furthermore, nodes are anonymous such that nodes do not have IDs (identifiers) but each incident edge is uniquely identified by a labeled port number from $[1, \delta]$ where $\delta$ is the degree of the node. The nodes connected via an edge are termed neighboring nodes and are considered to be one hop away from each other. There are two port numbers assigned to any edge ($e$) corresponding to two nodes $u, v \in G$, that connect these two nodes. Moreover, there is no relationship between any two port numbers of an edge.\\

%, KMS19
\noindent{\bf Robots:}  The set of robots $R = \{ r_1, r_2, \dots, r_k\}$ represents a collection of $k \leq n$ robots that are located across the graph $G$ at one or more nodes. Robots do not stay at the edge and stay only on the nodes of the graph $G$. Two or more robots situated at a node are termed as co-located robots and these co-located robots can communicate with each other. This model is known as {\em local communication model} \cite{AM18, KA19}. On the other hand, if the robots are allowed to communicate with any other robots in the graph (need not be co-located), the model is known as the {\em global communication model} \cite{KMS22,KMS20}. However, in this paper, we consider only the local communication model. Each robot contains a unique ID consisting of $O(\log k)$ bits. A robot can move from one node to another from the port if the nodes are connected to each other via an edge. Each robot consists of some memory to store information and computation. We considered the computation time to be negligible as compared to the movement time of the robots from one node to another. A robot performs the "Communicate-Compute-Move" operation which is defined below.\\

\noindent{\bf Cycle:}   We consider a synchronous setting, every robot is synchronized to a common clock and movement from one node to another is complete in a one-time cycle or round. A robot $r_i \in R$ remains active in the "Communicate-Compute-Move" (CCM) cycle in a synchronous setting. Following are the three operations carried out by the robots:
\begin{itemize}
    \item {\bf Communicate:} $r_i$ can interact with the co-located robots and view the memory of a different robot, say $r_j \in R$.
    \item {\bf Compute:} $r_i$ can perform any required computation by using the data gathered during the "communicate" phase. This involves choosing a (potentially) port to leave $v_i$ and selecting the data to be saved in the robot $r_j$.
    \item {\bf Move:} $r_i$ writes new information (if any) in the memory of a robot $r_j$ at $v_i$,  and exits $v_i$ using the computed port to reach to a neighbor node of $v_i$. 
\end{itemize}

\noindent{\bf Crash Faults:} We consider the crash failure setup where a robot may fail by {\em crashing} at any time during the execution of the algorithm. The crashed robot is not recoverable and once a robot crashes it immediately loses all the information stored in itself, as if it was not present at all. Further, a crashed robot is not visible or sensible to other robots. We assume there are at most $f$ faulty robots such that $f\leq k$.\\

\noindent{\bf Time and Memory Complexity:} 
% We consider a synchronous system, where every robot is synchronized to a common clock. The movement of a robot from one node to another node completes in one-time cycle or round. 
The time complexity conveys the number of discrete rounds or cycles taken before achieving dispersion. Memory complexity is the number of bits required to store each robot to successfully execute the algorithm. Our goal is to solve dispersion as fast as possible and keep the memory per robot as low as possible.\\

Below, we formally define non-faulty and faulty robots' dispersion problems in a graph. We study the problems on an oriented as well as unoriented grid graph that too with and without fault in this paper.

\begin{definition}[\sc Non-Faulty Robots Dispersion]
\label{def:NFT-dispersion}
Given $k \leq n$ robots, initially, placed arbitrarily over the nodes of an $n$-node graph, the robots re-position themselves autonomously such that each node has at most one robot on it and subsequently terminate.
\end{definition}

\begin{definition}[\sc Faulty Robots Dispersion]
\label{def:FT-dispersion}
Given $k \leq n$ robots, up to $f\leq k$ of which are faulty, initially placed arbitrarily over the nodes of an $n$-node graph, the (non-faulty) robots re-position themselves autonomously such that each node has at most one non-faulty robot on it and subsequently terminate. 
\end{definition}

\section{Related Work}\label{sec: related_work}
In this section, we discuss the work related to deterministic algorithms for the dispersion of robots.

In 2018, Moses Jr. et al. introduced the dispersion very first time \cite{AM18}, where they presented dispersion for several types of graphs. They showed a lower bound of $\Omega(\log n)$ for each robot's memory. Later in 2019,  the lower bound was made more specific w.r.t. $k$ (number of robots) with $\Omega(\log (\max(k, \Delta)))$ in \cite{KMS19} where $\Delta$ is the highest degree of the graph. They showed that for an arbitrary graph, a deterministic algorithm requires $\Omega(D)$ rounds, where $D$ is the diameter of the graph. They also presented two techniques for the dispersion of robots on an arbitrary graph, one takes $O(\log n)$ memory and $O(mn)$ time while the other requires $O(n \log n)$ memory and $O(m)$ time.

In 2019, Kshemkalyani and Ali discussed several techniques for both synchronous and asynchronous models \cite{KA19}. They used $O(k \log \Delta)$ memory and $O(\min(m, k\Delta))$ rounds to solve the dispersion problem in the synchronous model. They presented many algorithms for the asynchronous model, one of which requires $O(\Delta D)$ rounds and $O(D \log \Delta)$ memory while another requires $O(\max(\log k,\log \Delta))$ memory with $O((m-n)k)$ time complexity. In the same year 2019, Kshemkalyani et al. improved the time complexity to $O(\min(m, k\Delta) \log k)$ at the cost of $O(\log n)$ memory \cite{KMS19}, but the knowledge of $m, n, k,$ and $\Delta$ were required by the robots in advance. In 2020, Takahiro et al. achieved the same complexity without the knowledge of the parameter $m, n, k,$ and $\Delta$ \cite{SSKM20}. Later in 2021, Kshemkalyani and Sharma improved the time complexity to $O(\min(m, k\Delta))$ with the graph-specific termination of the algorithm \cite{KS21}. %Randomization was employed in the works of \cite{MollaM19} and \cite{DBS21}, which helped to lower the amount of memory needed for each robot.

%Kshemkalyani et al. [1] explored the dispersion problem in the Global Communication Model, where robots can communicate with one another regardless of where they are located in the graph. They  required $O(\log (k + \Delta))$ bits of memory at each robot and achieved the time complexity of $O(k\Delta)$ rounds. In another setting, the number of rounds improved to $O(\min(m, k\Delta))$ when robots require $O(\Delta + \log k)$ bits of memory. These both settings work in the initial arbitrary configuration of the robots. They also explored dispersion using BFS traversal methods. For various robot beginning configurations, the BFS traversal technique produced a duration of $O((D + k) \Delta ((D + \Delta))$ rounds with $O(\log D + \Delta \log k)$ bits of memory at each robot, using global communication. Here, $D$ stands for the graph's diameter. Dynamic graphs were used to study the issue in \cite{KMS020,AAMKS18,LFPPSV20}. In-depth research on graph exploration is discussed in \cite{BGHIKK09,CFIKP08,DDKPU13,FIPPP05}.

In 2020, the dispersion problem is studied for faulty robot configurations. In \cite{MMM20}, Molla et al., dispersed the robots in  weak Byzantine settings (robots that behave randomly but cannot change their IDs) and examined the problem of unidentifiable rings. In \cite{MMM21} authors suggested several dispersion methods, some of which were tolerant of strong Byzantine robots (robots that behave arbitrarily and can even modify their IDs). Their algorithms are primarily based on the concept of assembling the robots at a root vertex, using them to create an isomorphic map of $G$, and then scattering them throughout $G$ by a predetermined protocol. %In \cite{PS021}, dispersion under crash faults was addressed. The issue of a group of robots starting at a rooted configuration with certain robots being crash-prone has been discussed by Pattanayak et al. in \cite{PS021} where each robot needs $O(\log (k + \Delta))$ bits of memory and can handle any number of crashes. The algorithm takes $O(f \cdot \min(m, k\Delta))$ rounds to complete the dispersion.
Recently, Chand et. al. \cite{CKMS23} provided two algorithms on the mobile robots' dispersion on arbitrary graphs in the presence of crash fault robots. The first algorithm which has a rooted initial configuration has a time complexity of $O(k^2)$. On the other hand, the second algorithm which has an arbitrary initial configuration has a time complexity of $O((f + l) \cdot \min(m, k\Delta, k^2))$  when all the parameters are known to the robots, where $l$ is the cluster of robots in initial configuration.
 
In the grid graph, the dispersion problem was studied by Kshemkalyani et al. in \cite{KMS20}. They consider an oriented (called planar) grid and only non-faulty robots. In an oriented grid a single robot can move through a path of the grid in the same direction using the ports' ordering at a node. This property makes it easier for them to solve the dispersion problem in $O(\min(k, \sqrt{n}))$ time with each robot requiring only $O(\log k)$ bits of memory in the local communication model. They also studied the problem in the global communication model, reducing the round complexity to $O(\sqrt{k})$ rounds. 
We study the dispersion problem on an unoriented grid, considering both faulty and non-faulty robots in the local communication model. A quick comparison of results is provided in Table~\ref{tab:results}.

\section{Dispersion of the Faulty Robots on Oriented Grid}\label{sec: faulty-planar}
In this section, we present a deterministic algorithm for the dispersion of faulty robots arbitrarily placed in the oriented grid, in which robots have a sense of direction. Our primary goal is to achieve the round complexity as close as possible to non-faulty robots' dispersion in the oriented square grid and keep the memory as low as possible. Specifically, we achieve the optimal round and memory complexity, i.e., $O(\sqrt{n})$ messages and $O(\log n)$ bits memory at each robot.

\subsection{Algorithm}
The dispersion algorithm designed for non-faulty robots \cite{KMS20} cannot be readily adapted for faulty robots in oriented grids. Due to faulty robots, dispersion is more challenging as compared to non-faulty robots in an oriented square grid. The main idea is to reach the corner of the grid and then disperse from there but finding a single corner, efficiently, for gathering is challenging due to the faulty nature of the robots. Therefore, based on the size of the grid robots meet in the middle or partition the whole grid into four parts and disperse. Let us call this algorithm as {\em Oriented Grid Dispersion with Faulty Robots}. Below, we explain the algorithm in detail by first gathering the robots at the corners and then based on the $n$ whether odd or even, executing the protocol.

\textbf{Gathering at the corners of the grid:} Each robot $r_u$ (with degree $4$) initiates from the minimum port number and reaches the boundary (node with degree $3$) of the grid by moving in a straight line.  It takes at most $\sqrt{n}$ rounds. Robot $r_u$ at the boundary node moves along the boundary of the grid by initiating the minimum port number and reaches the corner (node with degree $2$) of the grid, which takes at most $\sqrt{n}$ rounds. In $2\sqrt{n}$ rounds, all the robots gather across the four corners. Each robot requires $O(\log n)$ bits of memory to keep the account of number of rounds that have passed and their own ID.

\textbf{Dispersion for even number of node:} A corner keeps $n/4$ lowest ID robots at its corner and sends the remaining to another corner from the minimum port number along the boundary. If any corner receives some extra robots from some corner then send those robots to the opposite corner. To move over the boundary of the grid and access all the corners takes $3\sqrt{n}$ rounds in the worst case. Therefore, after $5\sqrt{n}$ rounds ($2\sqrt{n}$ rounds used to reach the corner), corner with $C_r$ robot sends $\frac{2C_r}{\sqrt{n}}$ (at most $\sqrt{n}/2$) robot in each column with the help of the boundary nodes (including itself) in the direction of the minimum port number, if available. Each robot $r_u$ takes at most $5.5 \sqrt{n}$ rounds to reach their respective column based on their IDs, lower the ID nearer the column number. In the next $\sqrt{n}/2$ rounds, each robot disperses on the grid in their respective column in the order of lower to higher ID. Therefore, it takes $6\sqrt{n}$ rounds to disperse the robots in the square grid when $n$ is even. In this way, a corner covers at most $1/4^{th}$ grid, i.e., $\sqrt{n}/2 \times \sqrt{n}/2$ to disperse the $n/4$ robots (if available) at those nodes. 

\textbf{Dispersion for an odd number of node:} After $2\sqrt{n}$ rounds, corner robots move to the middle of the boundary node by initiating in the direction of the minimum port number, i.e., $\lfloor\sqrt{n}/2\rfloor$ hop away from the corner. It takes $\sqrt{n}/2$ rounds. After $2.5\sqrt{n}$ rounds ($2\sqrt{n}$ rounds used to reach the corner), robots at the middle of the boundary move $\lfloor\sqrt{n}/2\rfloor$ hop away from the boundary node at $90\degree$ and reach the center of the grid in another $\sqrt{n}/2$ rounds. Since $n$ is odd, therefore, $\sqrt{n}$ is also odd. Notice that $\sqrt{n}$ is an integer. This implies, that there exists a unique node at equal distance between any two non-diagonal corner nodes. Furthermore, if any two robots are situated in the middle of the boundary node at opposite boundary then also there exists a unique middle node between them. Therefore, all the non-faulty robots meet at the center of the grid which is equidistant from all the corners. Henceforth, there exists a unique center node in the grid. After $3\sqrt{n}$ rounds, the centered robots collectively move to the corner of the grid as discussed above in \textit{gathering at the corners of the grid}. Notice that it takes $\sqrt{n}$ rounds to reach the corner of the grid from the center of the grid since robots move in straight lines and reach the middle of the boundary in $\sqrt{n}/2$ rounds and further reach at the corner of the grid in next $\sqrt{n}/2$ rounds. After $4\sqrt{n}$ rounds, a single corner possesses $S_r$ robots and sends $\frac{S_r}{\sqrt{n}}$ robot in each column with the help of the boundary nodes (including itself) in the direction of the minimum port number, if available. Each robot $r_u$ takes at most $\sqrt{n}$ additional rounds to reach their respective corner based on their IDs, lower the ID nearer the column number. In the next $\sqrt{n}$ rounds, each robot is dispersed on the grid in their respective column in the order of lower to higher ID. Therefore, it takes overall $6\sqrt{n}$ rounds to disperse the robots in the square grid when $n$ is odd.

\begin{algorithm}\caption{\sc Oriented Grid Dispersion with Faulty Robots} \label{alg: faulty-oriented}
    \begin{algorithmic}[1]
    \Require{A square grid of $n = \sqrt{n}\times \sqrt{n}$ nodes, where $f$ robots are faulty such that$ f \leq k \leq n$, $k$ is the number of robots. The robots are distributed on the grid.}
    \Ensure{Dispersion of the robots over the nodes.}
    \Statex
    \State Each robot $r_u$ (with degree $4$) reaches at the boundary (node with degree $3$) of the grid by moving in straight line. Robot $r_u$ placed at boundary node, reaches the corner (node with degree $4$) of the grid. \Comment{Takes $O(\sqrt{n})$ rounds.} \label{line: oriented_at_the_corner}
    
    \If{$n$ is even}\Comment{takes $O(\sqrt{n})$ rounds.}
        \State  $r_u$ keeps $n/4$ robot at its corner and sends the remaining robots to another corner. 
        \State After $5\sqrt{n}$ rounds, corner with $C_r$ robots sends $\frac{2C_r}{\sqrt{n}}$ robots in half ($\sqrt{n}/2$) of the column, if available. \label{line: oriented_even_column}
    \EndIf
    
    \If{$n$ is odd}\Comment{Takes $O(\sqrt{n})$ rounds.}
        \State After $2\sqrt{n}$ rounds, corner robots move to the middle of the boundary, i.e., $\lfloor\sqrt{n}/2\rfloor$ hop away from the corner.%\Comment{Takes $\sqrt{n}/2$ rounds.}
        \State After $2.5\sqrt{n}$ rounds, robots at the middle of the boundary move $\lfloor\sqrt{n}/2\rfloor$ hop away from the boundary at $90\degree$ and reach the center of the grid. %\Comment{Takes $\sqrt{n}/2$ rounds.}
        \State After $3\sqrt{n}$ rounds, the centered robots move to the corner of the grid.
        \State After $4\sqrt{n}$ rounds, single corner possess $S_r$ robots and sends $\frac{S_r}{\sqrt{n}}$ robot in each column, if available.\label{line: oriented_odd_column}
    \EndIf
    \State Each column disperses the robots across the column if available.\Comment{Takes $O(\sqrt{n})$ rounds.}
    \State All the robots settled at a unique node.
    \end{algorithmic}
\end{algorithm}%\label{alg: faulty-planar}

From the above discussion, we have the following results.

\begin{theorem}
    Consider any oriented square grid of $n$ nodes having $f$ faulty robots such that $f \leq k \leq n$ and $k$ is the number of robots, in which, each robot has memory access $O(\log n)$ bits then DISPERSION can be solved deterministically in $O(\sqrt{n})$ rounds.
\end{theorem}
\begin{remark}
    Algorithm~\ref{alg: faulty-oriented} can be modified to work without the knowledge of $n$. After Line~\ref{line: oriented_at_the_corner} (in at most $\sqrt{n}$ rounds), each robot $r_u$ moves from one corner to another corner and learns the value of $\sqrt{n}$, so $n$. Therefore, the extra cost added to the Algorithm~\ref{alg: faulty-oriented} is of $\sqrt{n}$ round which keeps the overall round complexity unchanged, $O(\sqrt{n})$. 
\end{remark}

\subsection{Extension to Rectangular Grid}
We discussed the dispersion of faulty robots on the oriented square grid of $n = \sqrt{n} \times \sqrt{n}$ nodes in Algorithm~\ref{alg: faulty-oriented}. The Algorithm~\ref{alg: faulty-oriented} can be modified to the dispersion of the faulty robots on the rectangular grid of $n = \ell \times \frac{n}{\ell}$, where $\ell$ is the length and $n/\ell$ is the width of the rectangle such that $\ell, n/\ell>1$. In a square grid, length and width are of the same dimension, i.e., $\sqrt{n}$. In a rectangular grid, w.l.o.g, we consider that $\ell > n/\ell$. Henceforth, any operation that takes $\sqrt{n}$ rounds in a square grid would take at most $\ell$ rounds in a rectangular grid in every line of the Algorithm~\ref{alg: faulty-oriented}. After reaching at the corner (in line~\ref{line: oriented_at_the_corner}), it takes $\ell + n/\ell$ rounds to know the dimension of the grid. A robot $r_u$ can do this by moving from one corner to the diagonal corner. Therefore, the only change in the Algorithm~\ref{alg: faulty-oriented} is in the line~\ref{line: oriented_even_column} and line~\ref{line: oriented_odd_column} where $\frac{2C_r}{\ell}$ and $\frac{S_r}{\ell}$ robots are sent, respectively, across the longer column of the rectangular grid. This conveys that the rectangular grid's dispersion round and memory complexity would be $O(\ell)$ and $O(\log n)$, respectively. Since the diameter of the rectangular grid is $O(\ell)$, therefore, our algorithm is optimal w.r.t. round and memory complexity. The correctness and complexity proof directly follows from the square grid.
% \begin{remark}
%     For $n/\ell = 2$, robot can be dispersed on the boundary in $O(\ell)$ rounds.
% \end{remark}

\section{Dispersion of Non-Faulty Robots on Unoriented Grid}\label{sec: honest_robot}
In this section, we present a deterministic algorithm for dispersing $k\leq n$ robots initially placed arbitrarily on the nodes of an $n$-node unoriented square grid. The algorithm is both time and memory optimal, achieving dispersion within $O(\sqrt{n})$ rounds while utilizing $O(\log n)$ bits of memory for each robot. Note that all robots considered in this algorithm are non-faulty. 

\subsection{Algorithm}
%\shortOnly{
The dispersion algorithm designed for oriented grids \cite{KMS20} cannot be readily adapted for unoriented grids. %In the presence of non-faulty robots and square grids, there exists a dispersion algorithm but that one is in a oriented grid \cite{KMS20}. 
For unoriented grids, one can employ the DFS (Depth First Search) traversal based algorithm \cite{KS21} which works in an arbitrary graph and solves the dispersion problem in $O(\text{min}(m, k\Delta))$ rounds, where $m$ is the number of edges, $\Delta$ is the highest degree in the graph and $k\leq n$ is the number of robots. However, this algorithm takes $O(n)$ rounds, given that $m = O(n)$ and $\Delta = 4$ in a grid graph. Therefore, the challenge lies in reducing the time complexity from $O(n)$ to $O(\sqrt{n})$.
%}

    In the unoriented grid, the main challenge is the sense of direction as compared to the oriented grid as discussed before. Thus, a single robot situated at the internal nodes cannot move on a path in the same direction of the grid. However, we show that two or more robots working together can navigate along a path in the same direction. This ability to move along a path is crucial for achieving dispersion within $O(\sqrt{n})$ rounds since the length of a path in the square grid is $\sqrt{n}$. The outline of the algorithm is, first we gather the robots at a corner node of the grid. For this, a single robot at the node remain settle at its node. The robots in a group of two or more on the internal nodes (i.e., the node with degree four) move to the boundary nodes first and then move to a corner node. Then from the four corner nodes, the robots gathered at one corner node. Suppose the robots are gathered in the top-left corner of the grid. Notice that all the $k$ robots may not be gathered, as some singleton robots may be settled already. Then the algorithm sends some robots to each column in parallel to count the number of single robots settled at each row. The robots back and report the requirement of robots at each row to the gathered robots in the top-left corner. Then appropriate number of robots are sent to each column from the corner node. Then robots are dispersed along each row in parallel. Our approach takes $O(\sqrt{n})$ rounds to gather the non-singleton robots at a corner and another $O(\sqrt{n})$ rounds to disperse from the corner, including counting the singleton robots. We call this algorithm as \textit{Unoriented Grid Dispersion}. Below, we explain our algorithm the Algorithm~\ref{alg: non-faulty dispersion} in detail, breaking it into 3 stages.
    %Inside, edge, corner ----explain in the model with degree and specifications.

\begin{algorithm}%[ht!]
%\footnotesize
    \caption{\sc Unoriented Grid Dispersion}
    \begin{algorithmic}[1]
    \Require{A square grid of $n = \sqrt{n}\times \sqrt{n}$ nodes and $k \leq n$ robots are distributed across the grid nodes.}
    \Ensure{Dispersion of the robots over the nodes.}
    \Statex
    \Statex \textbf{Stage 1: Gathering at corners of the grid}
    \State Each robot $r$ count the round number.
    \While{round number $<59\sqrt{n}$}
        \If{robot $r$ is at the corner (node with degree 2)}
            \State Stay at corner.
        \ElsIf{robot $r$ is at the boundary node (node with degree 3)}
            \State Reach the corner node and stay there.
        \ElsIf{Non-Singleton robots at the internal node}
            \State Move in a direction and reach the boundary of the grid.\Comment{Movement of the internal robots is explained in the \textit{Straight line movement of the internal robots}.}
        \ElsIf{Singleton robot at the internal node}
            \State Stay there.
        \EndIf
    \EndWhile
    %\State Non-Singleton nodes move in a direction and reach the edge of the grid. 
    %\State Robot at the edge (nodes with degree 3) reach at the corner of the grid.
    %\If{round number $\geq$ \textbf{some rounds (would mention)}}
    \Statex \textbf{Stage 2: Gathering at a single corner}
    \While{round number $<77\sqrt{n}$} \Comment{  
 $59\sqrt{n}+18\sqrt{n}$ rounds required for Stage 1 and Stage 2.}
        \State Robots gathered (more than one) at the corner, send the highest ID robot via edges in one direction, and find the minimum robot's ID.
        \State All the corner robots gathered at the minimum robot's ID corner.
    \EndWhile
    
    \Statex \textbf{Stage 3: Dispersion from the gathered corner node}
    \If{number of robots at the corner is not more than the boundary and corner nodes}%\Comment{$\leq 4\sqrt{n}-4$}
        \State Disperse the robots along the boundary node and corner.
    \Else
        \State Send a pair of robots to each column and count the number of robots required by each column.
        \State Send the required number of robots to each column (if available) and disperse.
    \EndIf
    \State All the robots settled at a unique node.
    
    \end{algorithmic}
\end{algorithm}\label{alg: non-faulty dispersion}
\medskip 

\noindent  \textbf{Stage 1 - Gathering at corners of the grid:} Initially, robots are arbitrarily placed in the unoriented square grid. Nodes with degree four (internal nodes) having two or more robots start moving in a straight line and reach the boundary node (node with degree three) of the grid -- straight line movement of the internal robots is explained later. If robot $r$ is already at the boundary node of the grid, then $r$ moves across the boundary node of the grid and reaches the corner (node with degree two) -- the movement of the robot across the boundary nodes of the grid is explained later. Robots at the corner (node with degree two) remain at the corner. As discussed in the straight line movement of the internal robots and movement of the robot across the boundary node of the grid takes $O(\log n)$ memory (in bits).   

\medskip
\noindent \textbf{Stage 2 - Gathering at a single corner:} After $59 \sqrt{n}$ rounds (shown in Lemma~\ref{lem: inside_to_corner}), all the internal non-singleton robots, initially, placed at the node with degree four would gather at the corner. In between, all the robots at the boundary nodes would also gather at the corner. These robots would be distributed across the four corners (gathering might not be at a single corner). Now, we would gather them at, a single corner, the corner which has a minimum ID robot. All the corners that have more than one robot would send the highest ID robot available at that node across the boundary node of the grid and reach their initial corner after exploring the rest of the three corners. Notice that the minimum ID robot does not move from its position, therefore, each moving robot gets to know about the minimum ID robot and all the robots gather at that corner. These movements of the robots take place across the boundary nodes of the grid which take $3\sqrt{n}$ rounds to reach one corner to another (discussed in Lemma~\ref{lem: inside_to_corner}). In the gathering at a single corner, the robot needs to keep account of the boundary nodes' traversal and the minimum ID robot's hop distance. This takes $O(\log n)$ memory (in bits).

 \medskip 
   \noindent \textbf{Stage 3 - Dispersion from the single corner:} In $77\sqrt{n}$ rounds (from the Lemma~\ref{lem: inside_to_corner} and Lemma~\ref{lem: gathering_at_single_corner}), all the internal non-singleton robots and non-internal robots on the grid gather at the single corner.  If the number of robots at the corner is not more than $4\sqrt{n}-4$ then disperse the robots along the boundary node and corner of the grid since there is no robot at the boundary nodes of the grid. On the other hand, if the number of robots at the corner is more than $4\sqrt{n}-4$, then send a pair of robots to each column (except the boundary node of the grid) and count the number of robots required. The corner node considers the minimum port number at the boundary and other nodes in the direction as a column of the grid. The count can be done with the help of two robots by the straight movement of the robots in the column (explained in the straight line movement of the internal robots). These pairs of robots move in the column of the grid and know the number of robots required at each column. After having the round trip of their respective column, these pairs of robots consider the port number from where they entered their respective column to reach their initial corner position and move in a straight line across the boundary nodes. After knowing the required number of robots for each column, the required number of robots is sent to the corresponding column (if available) and dispersed. Notice that the first and last column requires $\sqrt{n}$ robots due to the unavailability of any robot on those columns. In this stage, pair robots keep the account of their column number and count of the required number of robots which takes $O(\log n)$ memory (in bits).

    \medskip 
   \noindent \textbf{Straight Line Movement of the internal Robots:} Let us consider there exists an internal node with $t$ number of robots such that $t>1$ and the minimum ID robot is $r_m$. Now, all the robots move across the minimum port number, i.e., port $1$,  in a round. Now we want to make sure that the robots move in a straight line, therefore, these robots should not move across the same edge that they just crossed (say, backward port). This implies three ports' choices remain in which these robots can move. Now, the maximum ID robots (say, $r$) take the next minimum available port and explore it till four port distances away from the $r_m$ in all the possible combinations (by giving priority to minimum port number) which are $3*3*3 = 27$ total different nodes. This round trip would take at most $2*27 = 54$ rounds from one hop away from the $r_m$ while moving one hop away from the $r_m$ takes $2$ more rounds. Therefore, the total rounds in this round trip are $56$. Notice that there exist precisely two ways for $r$ to reach the initial position (position where $r_m$ stays) of the robots. These two positions would be at $180\degree$ to each other. Further, let us consider two cases: (i) one of the ports from which the $r$ reaches $r_m$ is the backward port, then the other port would be the desired one in a straight line. Therefore, all the robots would proceed from that port in a straight line. (ii) None of the ports from which $r$ reached $r_m$ is the backward port, therefore, none of the ports among these two ports are at $180\degree$ of the backward port. As a consequence, one of the remaining ports is the backward port and the other one is at $180\degree$ of the backward port. In this way, we find the desired port to proceed in a straight line. Now, all the robots move in a straight line from the appropriate port and repeat the process till all of them reach the boundary node of the grid. Notice that some findings may take less than $56$ rounds when the robot $r$ explores the edge of the grid. Although we can consider that particular path to reach the boundary of the grid which is explored at $90\degree$ of our straight line direction, we would prefer to move in the straight line. Observe that it would not affect our overall round complexity but would keep the simple fact in consideration, i.e., the straight-line movement of the robots. Consequently, to move a single hop in a straight line takes at most $56$ rounds. In the case of memory, the maximum ID robot keeps account of at most $28$ ports to move one hop in a straight line and maximum ID and $r_m$, therefore, $O(\log k)$ (in bits) is used.

 \medskip 
   \noindent \textbf{Movement of the robot across the boundary nodes of the grid:} Let us consider a robot $r$ is placed at the boundary node of the grid. Robot $r$ moves one hop from the minimum port then there arise two cases that do not lead to the corner node: (i) robot $r$ reaches the boundary node of the grid or (ii) robot $r$ reaches the node whose degree is four. In the first case, $r$ explores the minimum port among the rest of the ports (excluding the port due to which $r$ reached this boundary node of the grid) and finds the boundary node of the grid. In the second case, $r$ comes back to the boundary node explores the next available minimum port (minimum unexplored port), and finds the boundary node. Similarly, $r$ repeatedly explores the boundary nodes and reaches the corner of the grid. Notice that the boundary nodes' degree is three and only one port leads to the non-side node of the grid, therefore, $r$ may select at most one non-side node of the grid before finding the correct boundary node of the grid. Therefore, robot $r$ takes at most $3$ rounds to move in a specific direction towards the corner of the grid. Also, the robot keeps an account of the last two traversed ports which takes constant memory.
\begin{lemma}\label{lem: inside_to_corner}
    Any boundary node robot takes at most $3\sqrt{n}$ rounds to reach the corner, while the group of robots (non-singleton robots) placed inside the grid reach the corner of the grid in at most $59 \sqrt{n}$ rounds.
\end{lemma}
\begin{proof}
    Non-singleton robots first move from the internal node to the boundary node of the grid. Further, they move from the boundary node of the grid to the corner of the grid. As we have seen in \textit{straight line movement of the internal robots,} that to move in a straight line there are at most $56$ rounds required to move one hop. Also, there exists at most $\sqrt{n}$ such hop in a straight line of the grid. Therefore, any group of robots requires $56\sqrt{n}$ rounds to move from the internal node to the boundary node of the grid. 

    Further, \textit{movement of the robot on the edge of the grid} takes $3$ round to move one hop in a specific direction, and there exist at most $\sqrt{n}-1$ such hop. Hence, the number of rounds required to move across the boundary nodes before reaching a corner node is at most $3\sqrt{n}$.

    Therefore, overall rounds required to reach the corner node by non-singleton robots from the inside of the grid are $3\sqrt{n} + 56\sqrt{n} = 59 \sqrt{n}$, i.e, $O(\sqrt{n})$ rounds.
\end{proof}
% \begin{lemma}
%     Any robot takes at most $O(\sqrt{n})$ rounds to reach to the corner node from the edge node. %$3\sqrt{n}$ 
% \end{lemma}
\begin{lemma}\label{lem: gathering_at_single_corner}
    Gathering all the corner robots at a single corner node takes at most $18\sqrt{n}$ rounds.
\end{lemma}
\begin{proof}
    Firstly, the robots with the highest ID at their respective corner take a round of the whole grid via boundary nodes and corner edges. This takes four sides' exploration of the corner robots. It takes $3\sqrt{n}$ round for one boundary of the grid (shown in Lemma~\ref{lem: inside_to_corner}). Therefore, four sides take $12\sqrt{n}$ rounds. In this way, each node gets to know about the minimum ID robot present at the corner. Now, that minimum ID robot can be at most two sides ($2\sqrt{n}$ hops) away from any corner. Therefore, to reach that particular corner takes at most $6\sqrt{n}$ rounds. Hence, overall rounds required to gather at a single corner are at most $18\sqrt{n}$.
\end{proof}

\begin{lemma}\label{lem: dispersion_from_corner}
     Robots gathered at a corner node take at most $118\sqrt{n}$ rounds to disperse the robots on the square grid.  
\end{lemma}
\begin{proof}
    Let us call that corner $C_0$ where all the robots gather. if the number of robots is not more than $4\sqrt{n}-4$ then these robots would disperse across the boundary nodes which takes at most $12\sqrt{n}$ rounds, i.e, $O(\sqrt{n})$ rounds.

    On the other hand, if the number of robots at $C_0$ is more than $4\sqrt{n}-4$ then the robots in pairs explore all the columns. To reach a particular column take $3\sqrt{n}$ (due to the movement of the robot across the boundary node of the grid) and then to take the round trip of a particular column take $2*56 \sqrt{n}$ rounds. After that reporting about the required number of robots at a particular column to $C_0$  takes at most $3\sqrt{n}$ rounds. Now these robots reach a particular column in at most $3\sqrt{n}$ rounds and settle at their respective node in $56 \sqrt{n}$. This implies the total number of rounds is $(3+112+3)\sqrt{n} = 118\sqrt{n}$.
\end{proof}
\begin{lemma}
    The Algorithm~\ref{alg: non-faulty dispersion} takes $195 \sqrt{n}$ rounds and $O(\log n)$ bits of memory at each robot.
\end{lemma}
\begin{proof}
    From reaching the corner, gathering at a single corner, and dispersing all the robots in Lemma~\ref{lem: inside_to_corner}, Lemma~\ref{lem: gathering_at_single_corner} and Lemma~\ref{lem: dispersion_from_corner}, respectively, take $59\sqrt{n}+18\sqrt{n}+118\sqrt{n} = 195 \sqrt{n}$ rounds. Similarly, Stages $1, 2$, and $3$ show total memory required is $O(\log n)$ bits.
\end{proof}
From the above discussion, we conclude the following results.
\begin{theorem}
    Consider any unoriented square grid of $n$ nodes having $k$ robots such that $k \leq n$ where each robot has the memory $O(\log n)$ (in bits) then DISPERSION can be solved deterministically in $O(\sqrt{n})$ rounds.
\end{theorem}

\subsection{Extension to Rectangular Grid}
We discussed the dispersion of non-faulty robots on the square grid of $n = \sqrt{n} \times \sqrt{n}$ nodes in Algorithm~\ref{alg: non-faulty dispersion}. The Algorithm~\ref{alg: non-faulty dispersion} can be modified to the dispersion of the non-faulty robots on the rectangular grid of $n = \ell \times \frac{n}{\ell}$, where $\ell$ is the length and $n/\ell$ is the width of the rectangle, such that $\ell, n/\ell>2$. For $n/\ell = 2$, a singleton robot can not decide the direction of the movement on the boundary since all the neighboring nodes have the same degree 3. In a square grid, length and width are of the same dimension, i.e., $\sqrt{n}$. In a rectangular grid, w.l.o.g, we consider that $\ell > n/\ell$. Henceforth, any operation that takes $\sqrt{n}$ rounds in a square grid would take at most $\ell$ rounds in a rectangular grid. Therefore, we can modify the Algorithm~\ref{alg: non-faulty dispersion} for the rectangular grid by replacing the $\sqrt{n}$ to $\ell$. This conveys that the rectangular grid's dispersion round and memory complexity would be $O(\ell)$ and $O(\log n)$. The farthest distance between two nodes in a rectangle is $\ell + n/\ell$ hopes, i.e., the diameter of the rectangle. Hence, our modified algorithm for rectangular grids has optimal round complexity and memory complexity. The correctness and complexity proof directly follows from the square grid. For termination, knowledge of $\ell$ is required.
\begin{remark}
    For $n/\ell =2$, the dispersion can be solved in a rectangular grid if the values of $n$ and $\ell$ are known. In that case, each boundary robot settles at the empty node, and if there is no empty node then movement in a straight line is executed similar to the movement on the internal node, to break the symmetry.
\end{remark}

\section{Dispersion of Faulty Robots on Unoriented Grids}\label{sec: crash-fault non-planar}
In this section, we present a deterministic algorithm for the dispersion of faulty robots, arbitrarily placed on the nodes of a unoriented grid. The algorithm achieves dispersion within $O(\sqrt{n}\log n)$ rounds while utilizing  $O(\sqrt{n}\log n)$ bits of memory for each robot. 

%Our goal is to achieve time and memory efficient algorithm in this setup. Specifically, we want to achieve the round complexity $\Tilde{O}(\sqrt{n})$, where $\Tilde{O}$ hides the $O(\log^c n)$ factor for some constant $c>0$.

\subsection{Algorithm}
%\shortOnly{
In the presence of faulty robots, dispersion becomes more challenging. We can not adopt the procedure discussed in Section~\ref{sec: honest_robot}. The reason behind this is: Firstly, in the case of straight-line movement of the internal robots, there might be the case that there are $O(n)$ robots present at the internal node. And after constant rounds, the highest ID robot crash before deciding the direction of movement. This may happen $O(n)$ times, therefore, the round complexity's order becomes linear. Secondly, after reaching the corner, it would be costlier (in terms of rounds) to find the minimum ID robot at the corner of the grid. Since there might be the case that every robot that goes for the round trip via boundary nodes of the grid crashes every time. In the worst case, there might be $O(n)$ robots at the corner and each failure would cost $\Omega(\sqrt{n})$ rounds. Therefore, their consecutive failure would cost sub-quadratic round complexity. Thirdly, after gathering at a single corner (non-singleton internal robots and non-internal robots), it would be costlier to send a pair of robots to the columns and count the required number of robots in each column. If $O(n)$ robots are gathered at a single corner then sending them in pair and waiting for the $O(\sqrt{n})$ rounds would further take $O(\sqrt{n})$ rounds extra, in worst case. Each round would be $O(\sqrt{n})$ rounds costly. Consequently, round complexity would be linear.
%}

Our approach for the faulty robots' dispersion is similar to the approach, in 3 stages, of the non-faulty robots' dispersion. %\shortOnly{
We deal with the above-discussed challenges in detail and improve our algorithm stage-by-stage.
%}

\begin{algorithm}%[H]
%\footnotesize
    \caption{\sc Unoriented Grid Dispersion with Faulty Robots}
    \begin{algorithmic}[1]
    \Require{A square grid of $n = \sqrt{n}\times \sqrt{n}$ nodes and $k$ is the number of robots. The robots are distributed in the grid.}
    \Ensure{Dispersion of the robots over the nodes.}
    \Statex
    \Statex \textbf{Stage 1 - Gathering at corners of the grid:}
    \State Each robot $r$ count the round number (say $Counter_r$).
    \State Initially, $Counter_r = 0$.
    \While{$Counter_r<(56\log n + 59\sqrt{n})$}
        \If{robot $r$ is at the corner of the grid (node with degree 2)}
            \State Stay at corner.
        \ElsIf{robot $r$ is at node having degree 3}
            \State Reach the corner node and stay there.
        \ElsIf{Non-Singleton robots at the node having degree 4}
            \State Move in a straight direction and reach the boundary of the grid.\Comment{Movement of the robots is explained in the \textit{Straight Line Movement of the Faulty Robots}.}
        \ElsIf{Singleton robot $r$ at the node with degree 4}
            \State Stay there.
        \EndIf
    \EndWhile
    \Statex
    \Statex \textbf{Stage 2 - Gathering at a single corner:}
    \State Robots gathered (more than one) at the corner (say, $C_r$ robots).
    \For{$12\sqrt{n}\log n+6\sqrt{n}$ rounds}
        \State  $C_r$ robots send $\lceil C_r/2 \rceil$ higher ID robot (say, seeker robots) via boundary edges in minimum port direction to find the corner of minimum ID robot.
        \If{seeker and non-seeker robots meet with the same minimum robot's ID after round trip}
            %\State Consider the minimum robot's ID known and 
            \State All the corner robots (seeker and non-seeker) gathered at the minimum robot's ID corner.
            %\State Raise alarm.
        \Else
            \State Update $C_r$.
        \EndIf
    \EndFor
    \Statex
    \Statex \textbf{Stage 3 - Dispersion from the single corner:}
    \State Let the number of robots at the corner be $C_{r_m}$. 
    \State Consider the number of columns which does not require any robot as $R_0$. Initially, $R_0= 0$.
    \For{$236 \sqrt{n}\log n$  rounds} \Comment{See in Lemma~\ref{lem: faulty_corner_dispersion}.}
        \If {$C_{r_m} \leq$ number of boundary and corner nodes}
        \State Disperse the robots along the boundary and corner of the grid based on robots' ID.
            \State Exit the For Loop.
        \Else
            \State Send $\Big\lfloor \frac{C_{r_m}}{\sqrt{n}-R_0-2} \Big\rfloor$ robots to a column that requires robots (except the boundary edges of the grid) and disperse them in that column.\label{line: 26}
            \State The remaining robots report at the corner about the column which does not require robots.
            %\State Send the required number of robots to the $\sqrt{n}-R_0-2$ column (if available) and disperse.
            \State Update $C_{r_m}$ and $R_0$.
        \EndIf
    \EndFor
    \State All the robots settled at a unique node.
    
    \end{algorithmic}
\end{algorithm}\label{alg: faulty dispersion}

\medskip 

\noindent\textbf{Stage 1 - Gathering at corners of the grid:} Initially, robots are arbitrarily placed in the unoriented square grid. Node with degree four (internal node) having two or more robots start moving in a straight line and reach the boundary node (node with degree three) of the grid -- straight line movement of the internal faulty robots is explained later. If the robot $r$ is already at the boundary node of the grid, then $r$ moves across the boundary edge of the grid and reaches the corner (node with degree two) as explained in the movement of the robot across the boundary edges of the grid (see Section~\ref{sec: honest_robot}).  During the exploration of the boundary nodes and reaching the corner, faulty robots might crash and not reach the corner, unlike non-faulty robots.  Robots at the corner (node with degree two) remain at the corner.

\medskip 

\noindent\textbf{Stage 2 - Gathering at a single corner:} After $56\log n + 59\sqrt{n}$ rounds (shown in Lemma~\ref{lem: fauly_inside_to_corner}), all the internal non-singleton non-crashed robots, initially, placed at the node with degree four would gather at the corner. In between, all the robots at the boundary of the grid would also gather at the corner. These robots would be distributed across the four corners (gathering might not be at a single corner). After that, robots gather at a single corner that has a minimum ID robot known as whether crashed or non-crashed. All the corners which have more than two robots would send half of the higher ID robots (say, seeker robots) available at that node across the boundary edge of the grid from the minimum port number and reach their initial corner after exploring the rest of the three corners. The seeker robots and non-seeker robots of a corner node move at the corner which has/had the minimum ID robot according to both seeker and non-seeker robots. There might be a case that either all the  seeker robots or non-seeker robots crashed, in that case, the remaining robots of a corner repeat the process to find the minimum ID robot on the boundary of the grid. This would be repeated until the seeker robots sent from the minimum port numbers meet the non-seeker robots after a round trip with the same known minimum ID robot, or at most one robot is remaining at the corner. These movements of the robots take place across the boundary edges of the grid. The boundary edge robots of the grid take $3\sqrt{n}$ rounds to reach the corner and there would be $O(\log n)$ round trip in worst case. Therefore, round complexity for this stage is $O(\sqrt{n}\log n)$ (as shown in Lemma~\ref{lem: faulty_single_corner_gathering}).

\medskip 

\noindent\textbf{Stage 3 - Dispersion from the single corner:} In $12\sqrt{n}\log n+65\sqrt{n}+56\log n$ rounds (from Lemma~\ref{lem: fauly_inside_to_corner}, and Lemma~\ref{lem: faulty_single_corner_gathering}), all the internal non-singleton robots and non-internal robots on the grid would gather at the single corner. Let us consider $C_{r_m}$ as the number of robots at the single corner and $R_0$ as the number of columns requiring no robots. Initially, $R_0$ is 0.  If the number of robots at the corner is not more than $4\sqrt{n}-4$ then disperse the robots along the boundary edges and corner of the grid based on their ID (smaller the ID nearer the position) since there is no robot on the boundary nodes of the grid. On the other hand, if the number of robots at the corner is more than $4\sqrt{n}-4$, then send $\Big\lfloor \frac{C_{r_m}}{\sqrt{n}-R_0-2} \Big\rfloor$ robots to a column that requires robots (except the boundary edges of the grid) and disperse them in that column. The remaining robots report at the corner about the column which does not require robots and update $C_{r_m}$ and $R_0$. This process is repeated until less than $4\sqrt{n}-4$ robots are left, which would be dispersed across the boundary edges of the grid. 

 \medskip 

\noindent\textbf{Straight Line Movement for the internal Faulty Robots:} We traverse similarly with the faulty robots as we did with the non-faulty robots. As discussed earlier there might be the case that maximum ID robots might crash every time before reporting the appropriate port number (port in the same direction at $180\degree$) for the traversal. It might increase the round complexity up to $O(n)$. Therefore, we consider half of the robots for finding the appropriate port number as compared to the maximum ID robot, and the other half stay at their place, waiting for the appropriate port number. If only one of the robots survived to report the appropriate port number then we know the appropriate port. Even if all the robots of a group crashed, then at most half of the robots would be left. Further, the remaining half of the robots partition themselves into two groups and repeat the process. Therefore, without moving on the right port (appropriate direction) these robots might crash $O(\log n)$  times, otherwise they will move in the appropriate direction (in a straight line). Consequently, round complexity would be $O(\sqrt{n}+\log n) = O(\sqrt{n})$. In other words, let us consider there exists an internal node with $t$ number of robots such that $t>1$ and the minimum ID robot known is $r_m$. Now, all the robots move across the minimum port number, i.e., port $1$, and take one round. Now we want to make sure that the robots move in a straight line, therefore, these robots should not move across the same port which they just crossed (say, backward port). This implies there are three ports' choices remaining in which these robots can move. Now, $\lfloor t/2 \rfloor$ maximum ID robots take the next minimum available port and explore it till four port distances away from the group of $r_m$ robots in all the possible combinations (by giving priority to minimum port number) which are $3*3*3 = 27$ total different nodes. This round trip would take at most $2*27 = 54$ rounds from one hop away from the $r_m$ while moving one hop away from the $r_m$ takes $2$ more rounds. Therefore, the total rounds in this round trip are $56$. Notice that there exist exactly two ways for the group of $t/2$ robots (even a single robot reporting would be enough to know the appropriate port number) to reach the initial position (position where other $t/2$ robots stay which may crash) of the robots. These two positions would be at $180\degree$ to each other. Further, let us consider two cases: (i) one of the ports from which the $t/2$ robots reached the initial position (where the other half of the robots are waiting for the appropriate port number) is the backward port, then the other port would be the desired one in a straight line. Therefore, all the non-crashed robots would proceed from that port in a straight line. (ii) None of the ports is the backward port, then none of the ports among these two ports are at $180\degree$ of the backward port. Therefore, one of the remaining ports is the backward port and the other one is at $180\degree$ of the backward port. In this way, we find the desired port to proceed in a straight line. Now, all the robots move in a straight line from the appropriate port and repeat the process till all of them reach the boundary edge of the grid. There might be two more possibilities with faulty robots based on the crash that $t/2$ robots crashed without reporting the port at $180\degree$ or the other half crashed, in that case, repeat the process by doing the partition of the remaining robots in two halves. There are only $\log n$ such partitions possible.  Hence, to move a single hop in a straight line takes at most $56$ rounds or at least half of the robots would crash. Therefore, it would take at most $56(\log n + \sqrt{n})$ rounds in reaching the boundary edge of the grid.

\begin{lemma}\label{lem: fauly_inside_to_corner}
    In stage 1%of the Algorithm~\ref{alg: faulty dispersion}
    , the group of robots (non-singleton non-crashed robots) placed inside the grid reach corners of the grid in at most $56\log n + 59\sqrt{n}$ rounds.
\end{lemma}
\begin{proof}
       There are $54 \sqrt{n}$ rounds required if both the groups do not crash fully in the straight line movement for the internal faulty robots. But if one of the groups crashes fully, that might happen at most $\log n$ times. Hence, $56$ rounds can be wasted without finding the appropriate direction $\log n$ times. Therefore, the total number of rounds required to reach the boundary edge of the grid is $56(\log n + \sqrt{n})$. In Section~\ref{sec: honest_robot}, we have seen that a non-faulty (robot which does not crash) single robot takes $3\sqrt{n}$ robot in the movement of the robot across the boundary edge of the grid. Therefore, non-singleton non-crashed robots reach at corners of the grid in at most  $56\log n + 59\sqrt{n}$ rounds from the internal node of the grid.
\end{proof}

\begin{lemma}\label{lem: faulty_single_corner_gathering}
    In stage 2, %  of the Algorithm~\ref{alg: faulty dispersion}, 
    corner robots gather at a single corner in $12\sqrt{n}\log n+6\sqrt{n}$ rounds.
\end{lemma}
\begin{proof}
     Firstly, half of the robots with the highest ID at their respective corner takes a round of the whole grid via boundary edges and corner edges. This takes four boundary edge explorations for the corner robots. It takes $3\sqrt{n}$ round for one boundary of the grid (discussed in Lemma~\ref{lem: fauly_inside_to_corner}). Therefore, four boundaries take $12\sqrt{n}$ rounds. In this way, each robot gets to know about the minimum ID robot present at the corner. On the other hand, there might be the case that either half of the robots crash fully, then the remaining half (at most) of the robots take another round trip across the boundary edges of the grid. There can be at most $\log n$ such incidents. Therefore, the overall round cost might be $12\sqrt{n}\log n$ to know the minimum ID robot's corner. Secondly, the minimum ID robot can be at most two boundaries ($2\sqrt{n}$ hops) away from any corner. Therefore, to reach the particular corner takes $6\sqrt{n}$ rounds. Hence, overall rounds required to gather at a single corner are $12\sqrt{n}\log n+6\sqrt{n}$.
\end{proof}

\begin{lemma}\label{lem: correctness_faulty_single_corner_gathering}
    In Lemma~\ref{lem: faulty_single_corner_gathering}, all the non-crashed corner robots gather at the single corner.
\end{lemma}
\begin{proof}
    We prove the lemma by contradiction. Let us suppose there exist two corners, $C_1$ and $C_2$ in Lemma~\ref{lem: faulty_single_corner_gathering} such that $C_2$'s corner ID was minimum then $C_1$. But the corner $C_0$ figures out the minimum ID robot at corner $C_1$ this implies that $C_0$'s seeker robots neither crossed the $C_2$'s non-seeker robots nor met in the middle with the seeker robots of the $C_2$. This implies there are no non-seeker robots at $C_2$. Additionally, seeker robots of $C_2$ did not cross the non-seeker robots of the $C_0$. Consequently, there are no seeker robots from the corner $C_2$ exist in the grid. Therefore, there do not exist any $C_2$ robots with minimum ID. Hence, the lemma. 
\end{proof}
\begin{lemma}\label{lem: faulty_corner_dispersion}
    In stage 3, % of the Algorithm~\ref{alg: faulty dispersion},
    dispersion from the single corner takes $O(\sqrt{n} \log n)$ rounds.
\end{lemma}
\begin{proof}
    In Stage 3, if the number of robots at the single corner is less than $4\sqrt{n}-4$  then these robots are placed across the boundary nodes of the grid in $4 \cdot 3\sqrt{n}$ rounds based on their ID. On the other hand, if the number of robots at the single corner is more than $4\sqrt{n}-4$ then an equal number of robots are sent in each column. There can be two cases: either half or more robots crash before reporting or not. If half or more robots report at their corner then half or more of the column does not require any robots further. Now, the number of columns left for the dispersion of the robots in the grid is less than or equal to half. This case can be repeated at most $\log n$ times, then there would be no column/node for the robots' dispersion. In another case, if less than half of the column reports to their corner then more than half of the robots are settled in the grid from the corner. This case also can be repeated at most $\log n$ times, after that there would be no robot for dispersion. In the worst case, these two cases can arise alternatively, where this process can be repeated at most $2\log n$ times. As we know before reporting, there are rounds required $6\sqrt{n}$ rounds for moving on the boundary edge and $112\sqrt{n}$ for moving in the column of the grid. Therefore, dispersion from a single corner takes $236 \sqrt{n} \log n$ rounds, i.e., $O(\sqrt{n}\log n)$.
\end{proof}
\begin{lemma}
    %Algorithm~\ref{alg: faulty dispersion}
    Our algorithm for faulty setup requires $O(\sqrt{n}\log n)$ memory at each robot.
\end{lemma}
\begin{proof}
    %In Algorithm~\ref{alg: faulty dispersion} 
    Each robot stores its own ID, therefore, $\log k$ bits are required. In stage 1  $O(\log k)$ memory is required to store the initial node's minimum ID and some constant memory to track the port number before deciding the appropriate straight direction for the internal robots. In the case of boundary edge movement, only port numbering is stored, therefore, constant memory is required. 

    In stage 2, traversal takes place along the boundary edges, therefore, the minimum ID corner and its position along with some constant port numbering require $O(\log n)$ memory.

    In stage 3, the distance from the corner and the number of robots required should be stored at each robot, along with the traversal at the boundary edge and inside the grid. Therefore, $2\sqrt{n}\log n$ bits are required to store the position of the column and the required number of robots. Consequently, all these stages require the $O(\sqrt{n}\log n)$ bits memory.  
\end{proof}
From the above discussion, we have the following results.

\begin{theorem}
    Consider any unoriented square grid of $n$ nodes in the presence of any number of faulty robots where $k$ is the number of robots, in which, each robot has access to memory $O(\sqrt{n}\log n)$ bits then DISPERSION can be solved deterministically in $O(\sqrt{n}\log n)$ rounds.
\end{theorem}

\subsection{Extension to Rectangular Grid}
We discussed the dispersion of faulty robots on the square grid of $n = \sqrt{n} \times \sqrt{n}$ nodes in Algorithm~\ref{alg: faulty dispersion}. The Algorithm~\ref{alg: faulty dispersion} can be modified to the dispersion of the faulty robots on the rectangular grid of $n = \ell \times \frac{n}{\ell}$, where $\ell$ is the length and $n/\ell$ is the width of the rectangle such that $\ell, n/\ell>2$. For $n/\ell = 2$,
a singleton robot can not decide the direction of the movement on the boundary since all the neighboring
nodes have the same degree 3. In a square grid, length and width are of the same dimension, i.e., $\sqrt{n}$. In a rectangular grid, w.l.o.g, we consider that $\ell > n/\ell$. Henceforth, any operation that takes $\sqrt{n}$ rounds in a square grid would take at most $\ell$ rounds in a rectangular grid in Stage 1, Stage 2, and Stage 3 of the algorithm~\ref{alg: faulty dispersion}. In Stage 3, line~\ref{line: 26} considers the longer edge with length $\ell$ for the dispersion of the robot across the columns. Therefore, we can modify the Algorithm~\ref{alg: faulty dispersion} for the rectangular grid by replacing the $\sqrt{n}$ to $\ell$. This conveys the rectangular grid's dispersion round and memory complexity would be $O(\ell \log n)$ and $O(\ell \log n)$. The correctness and complexity proof directly follows from the square grid. For termination, knowledge of $\ell$ is required.

\begin{remark}
    For $n/\ell =2$, the dispersion can be solved in a rectangular grid if the values of $n$ and $\ell$ are known. In that case, each boundary robot settles at the empty node, and if there is no empty node then movement in a straight line is executed similar to the movement on the internal node, to break the symmetry.
\end{remark}
%The farthest distance between two nodes in a rectangle is $\ell + n/\ell$ hopes, i.e., the diameter of the rectangle. Hence, our modified algorithm for rectangular grid has optimal round complexity and memory complexity.
 
\section{Conclusion and Future Work}\label{sec: conclusion}
In this paper, we studied dispersion for distinguishable mobile robots on a port-labeled square grid in an arbitrary configuration: oriented with fault and unoriented with and without fault.  We presented oriented and unoriented grids with and without fault, respectively, having the round complexity $O(\sqrt{n})$ and $O(\log n)$ bits memory. In contrast, the unoriented grid with faulty robots required both round and memory (in bits) $O(\sqrt{n}\log n)$.
%We presented $O(\sqrt{n})$ round and $O(\log n)$ bits memory algorithm for faulty robots in an oriented grid. Furthermore, in oriented grid non-faulty robot took $O(\sqrt{n}\log n)$ round algorithm for the faulty robots setup. In the non-faulty case, robots required $O(\log n)$ bits memory while the faulty case required $O(\sqrt{n}\log n)$ bits. 
Some open questions that are raised by our work: i) What is the non-trivial lower bound in terms of $k$ when $k = o(\sqrt{n})$, for the round complexity in both setups by keeping the memory $O(\log n)$? ii) Is it possible to improve the result in the faulty setup similar to the non-faulty setup? iii) Finally, whether similar bounds hold without the prior knowledge of $n$. %in the presence of Byzantine failures.

\bibliographystyle{plain}% the mandatory bibstyle
\bibliography{reference}

\begin{thebibliography}{10}

\bibitem{AM18}
John Augustine and William~K. {Moses Jr.}
\newblock Dispersion of mobile robots: {A} study of memory-time trade-offs.
\newblock In {\em {ICDCN} 2018.}

\bibitem{barriere2011uniform}
Lali Barriere, Paola Flocchini, Eduardo Mesa-Barrameda, and Nicola Santoro.
\newblock Uniform scattering of autonomous mobile robots in a grid.
\newblock {\em International Journal of Foundations of Computer Science}, 22(03):679--697, 2011.

\bibitem{becha2007optimal}
Hanane Becha and Paola Flocchini.
\newblock Optimal construction of sense of direction in a torus by a mobile agent.
\newblock {\em International Journal of Foundations of Computer Science}, 18(03):529--546, 2007.

\bibitem{CKMS23}
Prabhat~Kumar Chand, Manish Kumar, Anisur~Rahaman Molla, and Sumathi Sivasubramaniam.
\newblock Fault-tolerant dispersion of mobile robots.
\newblock In {\em {CALDAM} 2023}.

\bibitem{HABFM02}
Tien{-}Ruey Hsiang, Esther~M. Arkin, Michael~A. Bender, S{\'{a}}ndor~P. Fekete, and Joseph S.~B. Mitchell.
\newblock Algorithms for rapidly dispersing robot swarms in unknown environments.
\newblock In {\em {WAFR} 2002.}

\bibitem{HABFM03}
Tien{-}Ruey Hsiang, Esther~M. Arkin, Michael~A. Bender, S{\'{a}}ndor~P. Fekete, and Joseph S.~B. Mitchell.
\newblock Online dispersion algorithms for swarms of robots.
\newblock In Steven Fortune, editor, {\em SCG 2003.}

\bibitem{IPS22}
Giuseppe~F. Italiano, Debasish Pattanayak, and Gokarna Sharma.
\newblock Dispersion of mobile robots on directed anonymous graphs.
\newblock In {\em {SIROCCO} 2022}.

\bibitem{KA19}
Ajay~D. Kshemkalyani and Faizan Ali.
\newblock Efficient dispersion of mobile robots on graphs.
\newblock In {\em {ICDCN} 2019.}

\bibitem{KMS20}
Ajay~D. Kshemkalyani, Anisur~Rahaman Molla, and Gokarna Sharma.
\newblock Dispersion of mobile robots on grids.
\newblock In {\em {WALCOM} 2020.}

\bibitem{KMS20-dynamic}
Ajay~D. Kshemkalyani, Anisur~Rahaman Molla, and Gokarna Sharma.
\newblock Efficient dispersion of mobile robots on dynamic graphs.
\newblock In {\em ICDCS 2020}.

\bibitem{KMS19}
Ajay~D. Kshemkalyani, Anisur~Rahaman Molla, and Gokarna Sharma.
\newblock Fast dispersion of mobile robots on arbitrary graphs.
\newblock In {\em {ALGOSENSORS} 2019.}

\bibitem{KMS22}
Ajay~D. Kshemkalyani, Anisur~Rahaman Molla, and Gokarna Sharma.
\newblock Dispersion of mobile robots using global communication.
\newblock {\em J. Parallel Distributed Comput.}, 2022.

\bibitem{KS21}
Ajay~D. Kshemkalyani and Gokarna Sharma.
\newblock Near-optimal dispersion on arbitrary anonymous graphs.
\newblock In {\em {OPODIS} 2021.}

\bibitem{MMM20}
Subhrangsu Mandal, Anisur~Rahaman Molla, and William K.~Moses Jr.
\newblock Live exploration with mobile robots in a dynamic ring, revisited.
\newblock In {\em {ALGOSENSORS} 2020.}

\bibitem{MMM21}
Anisur~Rahaman Molla, Kaushik Mondal, and William~K. {Moses Jr.}
\newblock Byzantine dispersion on graphs.
\newblock In {\em {IPDPS} 2021.}

\bibitem{SSKM20}
Takahiro Shintaku, Yuichi Sudo, Hirotsugu Kakugawa, and Toshimitsu Masuzawa.
\newblock Efficient dispersion of mobile agents without global knowledge.
\newblock In {\em {SSS} 2020.}

\end{thebibliography}

%\appendix
%\section*{Appendix}
%\input{planar grid faulty robots}

\end{document}